\documentclass[11pt,twoside]{article}

\usepackage[usenames,x11names]{xcolor}

\usepackage[english]{babel}
\usepackage{latexsym,amssymb,amsmath,amsthm}
\usepackage{graphicx}
\usepackage{url}
\usepackage[pagebackref,breaklinks]{hyperref}
\usepackage{breakurl}
\usepackage{enumitem}
\usepackage[a4paper]{geometry}

\hypersetup{%
	pdftitle={Counting Triangulations and other Crossing-free Structures Approximately},%
	pdfauthor={Victor Alvarez and Karl Bringmann and Saurabh Ray and Raimund Seidel},%
	pdfkeywords={Algorithmic geometry, Algorithms, Geometry, Crossing-free structures, Counting algorithms, Triangulations, Matchings, Spanning Cycles, Tress, Approximation Algorithms},%
	pdfdisplaydoctitle=true,%
	pdfstartview=Fit,%
	colorlinks=false,%
	linktocpage,
}

\renewcommand*{\backrefalt}[4]{%
\ifcase #1 %
No citations.%
\or
One citation on page #2.%
\else
#3 citations on pages #2.%
\fi
}

\newtheorem{theorem}{Theorem}
\newtheorem{lemma}{Lemma}

\newcommand{\setp}{P}
\newcommand{\F}{\mathcal{F}}
\newcommand{\C}{\mathcal{C}}
\newcommand{\Conv}{\textup{Conv}}
\newcommand{\Approx}{\Lambda}
\renewcommand{\SS}{\mathcal{S}}
\newcommand{\CDTA}{\overline{\triangle}}
\newcommand{\CDT}{\textup{CDT}}
\newcommand{\CDTT}{\triangle}
\newcommand{\CFS}{S}
\newcommand{\SC}{\Delta}
\newcommand{\SP}{Q}
\newcommand{\nil}{\textup{nil}}

\title{Counting Triangulations and other Crossing-Free Structures Approximately}

\author{Victor Alvarez\thanks{Information Systems Group, Universit\"{a}t des Saarlandes, Germany, {\tt alvarez@cs.uni-saarland.de}.}
        \and
        Karl Bringmann\thanks{Max Planck Institute for Informatics, Karl Bringmann is a recipient of the \emph{Google Europe Fellowship in Randomized Algorithms}, and this research is supported in part by this Google Fellowship. {\tt kbringma@mpi-inf.mpg.de}.}
        \and 
        Saurabh Ray\thanks{Ben Gurion University of the Negev, Israel {\tt saurabh@math.bgu.ac.il}.}
        \and 
        Raimund Seidel\thanks{Fachrichtung Informatik, Universit\"{a}t des Saarlandes, Germany, \tt{rseidel@cs.uni-saarland.de}.}}

\index{Alvarez, Victor}
\index{Bringmann, Karl}
\index{Ray, Saurabh}
\index{Seidel, Raimund}

\begin{document}
\thispagestyle{empty}
\maketitle

\begin{abstract}
We consider the problem of counting straight-edge triangulations of a given set~$\setp$ of $n$ points in the plane. Until very recently it was not known whether the \emph{exact} number of triangulations of $\setp$ can be computed asymptotically faster than by enumerating all triangulations. We now know that the number of triangulations of $\setp$ can be computed in $O^{*}(2^{n})$ time~\cite{Alvarez:2013:SAA:2493132.2462392}, which is less than the lower bound of $\Omega(2.43^{n})$ on the number of triangulations of \emph{any} point set~\cite{Sharir:2011:DRT:1988090.1988484}. In this paper we address the question of whether one can approximately count triangulations in sub-exponential time. We present an algorithm with sub-exponential running time and sub-exponential approximation ratio, that is, denoting by~$\Approx$ the output of our algorithm, and by $c^{n}$ the exact number of triangulations of $\setp$, for some positive constant $c$, we prove that $c^{n}\leq\Approx\leq c^{n}\cdot 2^{o(n)}$. This is the first algorithm that in sub-exponential time computes a $(1+o(1))$-approximation of the base of the number of triangulations, more precisely, $c\leq\Approx^{\frac{1}{n}}\leq(1 + o(1))c$. Our algorithm can be adapted to approximately count other crossing-free structures on~$\setp$, keeping the quality of approximation and running time intact. In this paper we show how to do this for matchings and spanning trees.
\end{abstract}

\section{Introduction}

Let $\setp$ be a set of $n$ points on the plane. A crossing-free structure on $P$ is a straight-line plane graph with vertex set $P$. Examples of crossing-free structures include triangulations, trees, matchings, and spanning cycles, also known as the polygonizations of $\setp$, among others. Let $X$ denote a certain type of crossing-free structures and let $\F_{X}(\setp)$ denote the set of all crossing-free structures on $P$ of type $X$.One of the most intriguing problems in Computational Geometry is the following: given $\setp$ and a particular type of crossing-free structures $X$, how fast can we compute the cardinality of $\F_{X}(\setp)$? 

Among all kinds of crossing-free structures on $\setp$, triangulations are perhaps the most studied ones, so let us first focus on them. How fast can we compute the number of triangulations of $\setp$? For starters, although triangulations are the most studied, the problem of computing the number of \emph{all} triangulations of $\setp$ is, in general, neither known to be \#P-hard, nor do we know of a polynomial time algorithm to even approximate that number. When the point set is in convex position, an easy recurrence relation can be derived showing that the number of triangulations spanned by $n$ points in convex position is $C_{n-2}$, where $C_k$ is the $k$-th catalan number. F. Hurtado and M. Noy \cite{hurtado1997counting} were able to find further formulas to exactly compute the number of triangulations of ``almost convex sets''. Unfortunately, the approach of finding exact formulas does not seem to take us much further. 

For \emph{any} given set of points $\setp$ it is possible to enumerate all its triangulations in time proportional to its number of triangulations. This is because the \emph{flip graph}\footnote{The flip graph is a graph whose vertex set is the set of triangulations of $\setp$ and where there is an edge between vertices of the flip graph if both corresponding triangulations can be transformed into each other by flipping exactly one of their edges.} of triangulations is connected, see~\cite{lawson1972generation,Sibson01011978}. Thus any graph traversal algorithm like DFS or BFS can be used to enumerate the vertices of the flip graph of triangulations of $\setp$. One limitation of such traversal algorithms, however, is that the amount of memory used is proportional to the number of vertices in the graph - which is known to \emph{always} be exponential, because the number of triangulations of $\setp$ lies between $\Omega(2.43^n)$~\cite{Sharir:2011:DRT:1988090.1988484} and $O(30^n)$~\cite{DBLP:journals/combinatorics/SharirS11}. Using a general technique due to D. Avis and K. Fukuda called Reverse Search, see~\cite{Avis199621}, it is possible to enumerate triangulations while keeping the memory usage polynomial in $n$. This technique has been further improved in~\cite{Bespamyatnikh2002271,Katoh:2009aa}. However, it is important to observe that, since the number of triangulations is exponential in $n$, counting triangulations by enumeration takes exponential time, so a natural question is whether one can do significantly better. 

In~\cite{Aichholzer:1999:PT:304893.304896} O. Aichholzer introduced the notion of the ``path of a triangulation'', \mbox{T-path} from now on, that allowed a divide-and-conquer approach to speed up counting. From empirical observations, this approach seems to count triangulations in time sub-linear in the number of triangulations, that is, apparently faster than enumeration. Unfortunately, no proof that this algorithm is \emph{always} faster than enumeration - or even a good analysis of its running time - has been found. Subsequently, a simple algorithm based on dynamic programming was presented in~\cite{ray-seidel}. This algorithm empirically appears to be substantially faster than Aichholzer's algorithm. From the limited empirical data, the running time seems to be proportional to the square root of the number of triangulations. However, as with Aichholzer's algorithm, the worst case running time of this algorithms seems difficult to analyze and no bound on the running time has been found. More recently, two algorithms whose running times could actually be analyzed were presented. The first algorithm, presented in~\cite{sweep-line}, combines Aichholzer's idea of \mbox{T-paths} with a sweep-line algorithm and runs in time proportional to the largest number of \mbox{T-paths} found during execution (within a $poly(n)$ factor). In~\cite{sweep-line} the number of \mbox{T-paths} is shown to be at most $O(9^{n})$. It is important to observe that, for very particular and well-studied configurations of points, the number of T-paths can be shown to be significantly smaller than the number of triangulations. Thus, at least for those configurations, the algorithm of~\cite{sweep-line} counts triangulations faster than by using enumeration techniques. Unfortunately, this algorithm turned out to be very slow in practice, which is most probably due to the fact that the number of T-paths can still be very large, in~\cite{dumitrescu2001enumerating} a configuration having at least $\Omega\left(4^{n}\right)$ T-paths was shown. The second algorithm, presented in~\cite{Alvarez:2012:CCS:2261250.2261259}, uses a divide-and-conquer approach based on the onion layers of the given set of points. This algorithm was shown to have a running time of at most $O^{*}(3.1414^{n})$ and is likely to have a running time sub-linear is the number of triangulations since it is widely believed that the number of triangulations spanned by \emph{any} set of $n$ points is at least $\sqrt{12}^n \approxeq 3.46^n$, see~\cite{Santos2003186,aichholzer2001point,Aichholzer2004135}. From the experimental point of view, the second algorithm turned out to be significantly faster, for certain configurations of points, than the one shown in~\cite{ray-seidel}. These experiments can be found in~\cite{onion-layers}. Finally, in 2013, an algorithm with running time $O^{*}(2^{n})$ was presented~\cite{Alvarez:2013:SAA:2493132.2462392}. This last algorithm finally shows that enumeration algorithms for triangulations can indeed \emph{always} be beaten, as all point sets with $n$ points have at least $\Omega(2.43^n)$ triangulations. From an experimental point of view it was also shown to be significantly faster than all previous algorithms on a variety of inputs~\cite{Alvarez:2013:SAA:2493132.2462392}.

With respect to crossing-free structures other than triangulations the situation is very similar. In~\cite{Katoh:2009aa} a general framework for enumerating crossing-free structures (including spanning trees and perfect matchings) was presented. However, as for triangulations, the number of enumerated objects is again in general exponential. For example, the number of spanning trees is between $\Theta^{*}(6.75^{n})$~\cite{Flajolet1999203} and $O(141.07^{n})$~\cite{hoffmann2013counting}, and the number of perfect matching is between $\Theta^{*}(2^{n})$~\cite{garcia2000lower} and $O(10.05^{n})$~\cite{sharir2013counting}. The interested reader is referred to~\cite{aichholzer2006number, demaine, sheffer} for up-to-date lists of bounds for other crossing-free structures. Thus, again the question arises whether we can count these structures faster than by enumeration. In this respect, in~\cite{sweep-line} an algorithm was shown that counts pseudo-triangulations in time proportional to the largest number of pseudo-triangulation paths~\cite{aichholzer2003zigzag}, that the algorithm encounters. It is however still not known whether there are always considerably less pseudo-triangulation paths than pseudo-triangulations. In~\cite{razen2011counting} an algorithm was shown that counts \emph{all} crossing-free structures of a given set of points $\setp$ in time sub-linear in the number of counted objects, thus achieving the desired exponential speed-up in this case. Finally, it was until very recently that new fast counting algorithms appeared. In~\cite{Wettstein-MScThesis} it was shown, for example, that the number of \emph{all} crossing-free structures, perfect matchings, and convex partitions can be computed in time $O^{*}(2.839^{n})$ for the former and $O^{*}(2^{n})$ for the latter two. These algorithms generalize the idea presented in~\cite{Alvarez:2013:SAA:2493132.2462392} and show that enumeration can, at least in these cases, \emph{always} be beaten. It is, however, still open whether for spanning trees and spanning cycles, two of the most popular classes of crossing-free structures, the same can be said, see~\cite{Wettstein-MScThesis} for partial results in this direction.

\subsection{Our contribution}

The $O^{*}(2^{n})$ algorithm of~\cite{Alvarez:2013:SAA:2493132.2462392} for counting triangulations \emph{exactly} seems hard to beat at this point. 
Instead, we relax the goal to computing an \emph{approximation} of the number of triangulations and pose the question of whether one can reduce the runtime in this setting. The answer presented in this paper is, to the best of our knowledge, the first result in this new line of research. 

Note that, since for all sets of $n$ points the number of triangulations is $\Omega(2.43^n)$~\cite{Sharir:2011:DRT:1988090.1988484} and $O(30^n)$~\cite{DBLP:journals/combinatorics/SharirS11}, the quantity $\Theta\left( \sqrt{30 \times 2.43}^{\ n}\right)$ approximates the exact number of triangulations within a factor of $O\left( \sqrt{30/2.43}^{\ n}\right)$. Thus, one can trade the exponential time of an exact algorithm for a polynomial time algorithm with exponential approximation ratio. In this paper we bridge the gap between these two solutions by presenting an algorithm with sub-exponential running time and sub-exponential approximation ratio.

Let $\F_{X}(\setp)$ denote the set of all crossing-free structures of type $X$ on $\setp$, where $X$ could mean triangulations, matchings, or spanning trees. The main result of this paper is the following:

\begin{theorem}\label{theo:approxTri}
	Let $\setp$ be a set of $n$ points on the plane. Then a number $\Approx$ can be computed in time $2^{o(n)}$ such that $|\F_{X}(\setp)| \leq \Approx \leq |\F_{X}(\setp)|^{1+o(1)} = 2^{o(n)} |\F_{X}(\setp)|$.
\end{theorem}

The precise $o(n)$ terms mentioned in Theorem~\ref{theo:approxTri} are $O\left(\sqrt{n}\log(n)\right)$ for the running time and $O\left(n^{\frac{3}{4}}\sqrt{\log(n)}\right)$ for the approximation factor. At the end of~\S~\ref{sec:theoremTri} we mention a trade-off between these two, running time and approximation factor. 

While the approximation factor of $\Approx$ is rather big, the computed value is of the same order of magnitude as the correct value of $|\F_{X}(\setp)|$, that is, we compute a $(1+o(1))$-approximation of the base of the exponentially large value $|\F_{X}(\setp)|$. More precisely, if we denote $|\F_{X}(\setp)|$ by $c^{n}$, for some positive constant $c$ that depends on $\setp$ and $X$, then we have $c\leq\Approx^{\frac{1}{n}}\leq c^{1 + o(1)}\leq (1 + o(1))c$. Also, this approximation can be computed in sub-exponential time, which, at least theoretically, is asymptotically faster than the worst-case running times of the algorithms presented in~\cite{sweep-line,Alvarez:2012:CCS:2261250.2261259,Alvarez:2013:SAA:2493132.2462392,Wettstein-MScThesis}. This is certainly very appealing. 

The rest of the paper is divided as follows: We start in~\S~\ref{sec:preliminaries} with basic preliminaries. For simplicity, we first show in~\S~\ref{sec:count-tri} the algorithm for counting triangulations approximately and in~\S~\ref{sec:theoremTri} the corresponding proof of Theorem~\ref{theo:approxTri}. We generalize the algorithm for counting matchings and spanning trees in~\S~\ref{sec:count-ocfs}. We close the paper in~\S~\ref{sec:conclusions} with some remarks and conclusions.

\section{Preliminaries}\label{sec:preliminaries}

Our algorithm uses simple cycle separators as the main ingredient, originally presented in~\cite{Miller1986265} by G. L. Miller, and improved in~\cite{DBLP:journals/acta/DjidjevV97} by H. N. Djidjev and S. M. Venkatesan. The following theorem accounts for both results:

\begin{theorem}[Separator Theorem]\label{theo:miller}
	Let $T$ be a triangulation of a set of $n$ points in the plane such that the unbounded face is a triangle. Then there exists a simple cycle $C$ of size at most $\sqrt{4n}$, that separates the set $A$ of vertices of $T$ in its interior from the set $B$ of vertices of $T$ in its exterior, such that the number of elements of each one of $A$ and $B$ is always at most $\frac{2n}{3}$.
\end{theorem}

Observe that the Separator Theorem does not imply that \emph{every} triangulation of a set of points contains a \emph{unique} simple cycle separator. One can easily come up with examples in which a triangulation contains more than one simple cycle separator. The important part here is that \emph{every} triangulation contains \emph{at least} one simple cycle separator. 

\section{Counting triangulations approximately}\label{sec:count-tri}

The idea for an approximate counting algorithm is suggested by the Separator Theorem: We enumerate all possible simple cycle separators $C$ of size at most $\sqrt{4n}$ that we can find in the given set $\setp$. We then recursively compute the number of triangulations of each of the parts $A$ and $B$, specified by the Separator Theorem, that are delimited by $C$\footnote{Thus separator $C$ also forms part of the two sub-problems.}. We then multiply the number we obtain for the sub-problem $A\cup C$ by the number we obtain for sub-problem $B\cup C$, and we add these products over all cycle separators $C$. With this algorithm we clearly over-count the triangulations of $\setp$, and it remains to show that we do not over-count by too much. We will later see that in order to keep over-counting small, we have to solve small recursive sub-problems exactly. Note that problems of size smaller than a threshold $\SC$ can be solved exactly in time $O^{*}\left(2^{\SC}\right) = 2^{O(\SC)}$, see~\cite{Alvarez:2013:SAA:2493132.2462392}.

However, there are some technicalities that we have to overcome first. For starters, the Separator Theorem holds only if the unbounded face of $T$ is also a triangle. Thus, if we add a dummy vertex $v_{\infty}$ outside $\Conv(\setp)$, along with the adjacencies between $v_{\infty}$ and the vertices of $\Conv(\setp)$, to make the unbounded face a triangle, we can apply the Separator Theorem. Once a simple cycle with the dividing properties of a separator is found, by the deletion of $v_{\infty}$ we are  left with a separator that is either the original cycle that we found, if $v_{\infty}$ does not appear as a vertex of the separator, or a path otherwise. Thus, when guessing a separator we have to consider that it might be a path instead of a cycle.
This brings us to the second technical issue. As we go deeper in the recursion we might create ``holes'' in $\setp$ whose boundaries are the separators that we have considered thus far. That is, the recursive problems are polygonal regions, possibly with holes, containing points of $\setp$. Therefore, when guessing a separator, cycle or path, we have to arbitrarily triangulate the holes first. This does not modify the size of the sets we guess for a separator in a sub-problem.

We can now prove the first lemma:

\begin{lemma}\label{lemmas:1}
	Let $\F_{T}(\setp)$ be the set of triangulations of a set $\setp$ of $n$ points. Then \emph{all} separators, simple cycles or paths, among \emph{all} the elements of $\F_{T}(\setp)$ can be enumerated in time $2^{O\left(\sqrt{n}\log(n)\right)}$.
\end{lemma}
\begin{proof}
	We know by the Separator Theorem and the discussion beneath that \emph{every} element of $\F_{T}(\setp)$, a triangulation, contains at least one separator $C$, simple cycle or path. Moreover, the size of $C$ is at most $\sqrt{4n}$. Thus, searching by brute-force will do the job. We can enumerate all the sub-sets of $\setp$ of size at most $\sqrt{4n}$ along with their permutations. A permutation tells us how to connect the points of the guessed sub-set, after also guessing whether we have a path or cycle. We can then verify if the constructed simple cycle, or path, fulfils the dividing properties of a separator, as specified in the Separator Theorem. 
	
	It is not hard to check that the total number of guessed subsets and their permutations is $2^{O\left(\sqrt{n}\log(n)\right)}$. Verifying whether a cycle, or a path, is indeed a separator can be done in polynomial time. Thus, the total time spent remains being $2^{O\left(\sqrt{n}\log(n)\right)}$.
\end{proof}

We can now proceed with the corresponding proof of Theorem~\ref{theo:approxTri}.

\section{Proof of Theorem~\ref{theo:approxTri}}\label{sec:theoremTri}

We first prove that the approximation ratio of the algorithm for counting triangulations is sub-exponential and then we prove that its running time is sub-exponential as well.

\subsection{Quality of approximation}\label{c-tri:sections:approx:sub-sections:2}

By the proof of Lemma~\ref{lemmas:1} we also obtain that the number of simple cycle separators cannot be larger than $2^{O\left(\sqrt{n}\log(n)\right)}$. Since at \emph{every} stage of the recursion of the counting algorithm \emph{no} triangulation of $\setp$ can contain more than the total number of simple cycle separators found at that stage, we can express the \emph{over-counting factor} of the algorithm by the following recurrence:
\begin{align*}
	S(\setp, \SC) &\le \sum_{C} S(A\cup C, \SC)\cdot S(B\cup C, \SC)\leq 2^{O\left(\sqrt{n}\log(n)\right)}\cdot S(A\cup C^{*}, \SC)\cdot S(B\cup C^{*}, \SC)
\end{align*}
where the summation is over all separators $C$ available at the level of recursion.~$A\cup C$, $B\cup C$ are the sub-problems as explained before, $C^{*}$ is the cycle that maximizes the term $S(A\cup C, \SC)\cdot S(B\cup C, \SC)$ over all $C$, and $\SC$ is a stopping size. Specifically, whenever the current recursive sub-problem contains $\leq\SC$ points we stop the recursion and compute the number of triangulations of the sub-problem exactly. Hence, we have $S(P', \SC) = 1$ whenever $|P'|\leq\SC$. We can now write:
\begin{align*}
	S^{\prime}(\setp, \SC) := \log(S(\setp, \SC)) &\leq O\left(\sqrt{n}\log(n)\right) + S^{\prime}(A\cup C^{*}, \SC) + S^{\prime}(B\cup C^{*}, \SC).
\end{align*}

Our goal now is to prove the following lemma:

\begin{lemma}\label{lemmas:2}
  Let $P$ be a set of $n$ points on the plane and assume $\SC = n^{\Omega(1)}$, $n > \SC$, and $\SC$ is at least a sufficiently large constant. Then we have
	$$S^{\prime}(P, \SC) = O\left(\left(\frac{n}{\sqrt{\SC/3}} - \sqrt{n}\right)\log\SC\right).$$
\end{lemma}
\begin{proof}
	We use induction over the size of $P$. Let $P^{\prime}\subseteq P$ of size $m\leq n$. We have,
	\small
	\begin{align}
		S^{\prime}(P^{\prime}, \SC) &\leq O\left(\sqrt{m}\log(m)\right) + S^{\prime}(A\cup C^{*}, \SC)\ + S^{\prime}(B\cup C^{*}, \SC)\nonumber\\
		&\leq O\left(\sqrt{m}\log(m)\right) + c \left(\frac{m_{1}}{\sqrt{\frac{\SC}{3}}} - \sqrt{m_{1}} + \frac{m_{2}}{\sqrt{\frac{\SC}{3}}} - \sqrt{m_{2}}\right)\log\SC,\label{eq:1}
	\end{align}
	\normalsize
	where $m_{1}, m_{2}$ are the sizes of the sub-problems $A\cup C^{*}$ and $B\cup C^{*}$ of $P'$, respectively, and $c$ is some sufficiently large positive constant. By the Separator Theorem, we can express $m_{1}\leq\alpha m + \sqrt{4m}$ and $m_{2}\leq\beta m + \sqrt{4m}$, such that: (\oldstylenums{1}) $\alpha,\beta$ are constants that depend on the instance, so $\alpha = \alpha\left(A\cup C^{*}\right)$ and $\beta = \beta\left(B\cup C^{*}\right)$, (\oldstylenums{2}) $0 < \beta\leq\alpha\leq\frac{2}{3}$, and  (\oldstylenums{3}) $\alpha + \beta = 1$. 
	 
	 Now let us for the moment focus on the term $\frac{m_{1}}{\sqrt{\SC/3}} - \sqrt{m_{1}} + \frac{m_{2}}{\sqrt{\SC/3}} - \sqrt{m_{2}}$ of equation~(\ref{eq:1}) above:
	 \small
	 \begin{align*}
	 	\frac{m_{1}}{\sqrt{\frac{\SC}{3}}} - \sqrt{m_{1}} + \frac{m_{2}}{\sqrt{\frac{\SC}{3}}} - \sqrt{m_{2}} &= \frac{m_{1} + m_{2}}{\sqrt{\frac{\SC}{3}}} - \sqrt{m_{1}} - \sqrt{m_{2}}\\
		&\leq\frac{\alpha m + \sqrt{4m} + \beta m + \sqrt{4m}}{\sqrt{\frac{\SC}{3}}} -
		\sqrt{m_{1}} - \sqrt{m_{2}}\\
		&\leq\frac{m + 4\sqrt{m}}{\sqrt{\frac{\SC}{3}}} - \sqrt{\alpha m} - \sqrt{\beta m}\\
		&=\frac{m + 4\sqrt{m}}{\sqrt{\frac{\SC}{3}}} - \sqrt{m}\left(\sqrt{\alpha} + \sqrt{\beta}\right)\\
		&\leq\frac{m + 4\sqrt{m}}{\sqrt{\frac{\SC}{3}}} - \sqrt{m}\left(1 + \varepsilon\right)
	 \end{align*}
	 \normalsize
	 
	 The last inequality is obtained by minimizing $\sqrt{\alpha} + \sqrt{\beta}$. Since we mentioned before that $0\leq\beta\leq\alpha\leq\frac{2}{3}$ and $\alpha + \beta = 1$, the minimum of $\sqrt{\alpha} + \sqrt{\beta}$ is attained at $(\alpha, \beta) = \left(\frac{2}{3}, \frac{1}{3}\right)$, and is strictly larger than one, so we can choose $\varepsilon > 0$. Now, since $\SC$ is sufficiently large, we have $\frac{4\sqrt{m}}{\sqrt{\SC/3}}\ll\varepsilon\sqrt{m}$, so $\frac{4\sqrt{m}}{\sqrt{\SC/3}} - \varepsilon\sqrt{m} \le -\varepsilon' \sqrt{m}$, for some positive constant $\varepsilon'$. Thus we can continue as follows:
	 \small
	 \begin{align}
	 	\frac{m_{1}}{\sqrt{\frac{\SC}{3}}} - \sqrt{m_{1}} + \frac{m_{2}}{\sqrt{\frac{\SC}{3}}} - \sqrt{m_{2}} &\leq\frac{m + 4\sqrt{m}}{\sqrt{\frac{\SC}{3}}} - \sqrt{m}\left(1 + \varepsilon\right)\leq\frac{m}{\sqrt{\frac{\SC}{3}}} - (1+\varepsilon')\sqrt{m}.\label{eq:2}
	 \end{align}
	 \normalsize
	 
	 Combining equations~(\ref{eq:1}) and~(\ref{eq:2}) we obtain
	 \small
	 \begin{align*}
	 	S^{\prime}(P^{\prime}, \SC)&\leq O\left(\sqrt{m}\log(m)\right) + c\left(\frac{m}{\sqrt{\frac{\SC}{3}}} - (1+\varepsilon')\sqrt{m}\right)\log\SC\nonumber\\
		&\leq c\left(\frac{m}{\sqrt{\frac{\SC}{3}}} - \sqrt{m}\right) \log\SC + O\left(\sqrt{m}\log(m)\right) -
		c\cdot \varepsilon'\sqrt{m}\log\SC
	 \end{align*}
	 \normalsize
	 
	 If we choose $\SC$ to be sufficiently large, say $\SC\geq n^{\delta}$, for some constant $\delta > 0$, then we have $\SC\geq n^{\delta}\geq m^{\delta}$, and the negative term $-c\cdot \varepsilon'\sqrt{m}\log\SC$ is larger, for appropriately large $c$, than the $O\left(\sqrt{m}\log(m)\right)$ term. Hence, we can conclude that $S^{\prime}(P^{\prime}, \SC)\leq O\left(\left(\frac{m}{\sqrt{\SC/3}} - \sqrt{m}\right)\log\SC\right)$, which proves the induction step.
	 
	 It still remains to prove that the inductive claim holds for the boundary condition, so let $\SP$ be a recursive sub-problem of size $\leq\SC$. As $\SP$ stems from an application of the Separator Theorem, it is easy to see that $|\SP|\geq\frac{\SC}{3}$. Thus, we have $S^{\prime}(\SP, \SC) = 0\leq c\left(\frac{|\SP|}{\sqrt{\SC/3}} - \sqrt{|\SP|}\right)\log\SC$. Lemma~\ref{lemmas:2} follows.
\end{proof}

Now, let $\Approx$ be the number computed by our algorithm. Recall that $|\F_{T}(\setp)|$ is the exact number of triangulations of $\setp$. By setting $\SC = \sqrt{n}\log(n)$ we obtain an over-counting factor of the algorithm of: 
\begin{align*}
	S(P, \SC) = 2^{S^{\prime}(P, \SC)} = 2^{O\left(\frac{n\log\SC}{\sqrt{\SC}}\right)} = 2^{O\left(n^{\frac{3}{4}}\sqrt{\log(n)}\right)}
\end{align*}

Hence $|\F_{T}(\setp)|\leq\Approx\leq |\F_{T}(\setp)|\cdot 2^{O\left(n^{\frac{3}{4}}\sqrt{\log(n)}\right)} = |\F_{T}(\setp)|^{1 + o(1)}$. This completes the qualitative part of Theorem~\ref{theo:approxTri}. It remains to discuss the running time of the algorithm.

\subsection{Running time}\label{c-tri:sections:approx:sub-sections:1}

The running time of the algorithm can be expressed with the following recurrence:
\begin{align*}
	T(n) & < 2^{O\left(\sqrt{n}\log(n)\right)} \cdot T\left(\frac{2n}{3} + \sqrt{4n}\right).
\end{align*}

Taking again $T^{\prime}(n) = \log(T(n))$ yields $T^{\prime}(n) := T^{\prime}\left(\frac{2n}{3} + \sqrt{4n}\right) + O\left(\sqrt{n}\log(n)\right)$, which can then be solved using the well-known Akra-Bazzi Theorem for recurrences, see~\cite{leighton}. This yields $T^{\prime}(n) = O\left(\sqrt{n}\log(n)\right)$. There is, however, one detail missing, the stopping condition $\SC$. In order to use the Akra-Bazzi Theorem we need a boundary condition of $T(n) = 1$ for $1\leq n\leq n_{0}$ (for some constant $n_0$), but in the algorithm we stop the recursion whenever a sub-problem $\SP$ is of size $\leq\SC$ (which is dependent on the size $n$ of the original point set). At that point we compute the exact number of triangulations of $\SP$, which gives $T(|\SP|) = 2^{O(|\SP|)} = 2^{O(\SC)}$. Hence the exponent in the running time of the algorithm is upper-bounded by the solution of $T^{\prime}(n)$, as given by the Akra-Bazzi Theorem, plus $O\left(\SC\right)$, \emph{i.e.}, $T(n) = 2^{O\left(\sqrt{n}\log(n) + \SC\right)}$. If as before we assume that $\SC = \sqrt{n}\log(n)$ then we end up with $T(n) = 2^{O\left(\sqrt{n}\log(n)\right)} = 2^{o(n)}$, which concludes the proof of Theorem~\ref{theo:approxTri} for triangulations.

As a final remark observe that we could have used other values for $\SC$, rather than $\sqrt{n}\log(n)$, without violating any argument in the proofs. This yields a tradeoff with running time $2^{O(\SC)}$ and approximation ratio $2^{O\left(\frac{n \log\SC}{\sqrt{\SC}}\right)}$ for any $\sqrt{n} \log (n) \le \SC \le n$.
Although the quality of the approximation improves with larger $\SC$, the running time increases. Since we see no way of not having over-counting with this algorithm, we regard $\SC~=~\sqrt{n}\log(n)$ as the most reasonable setting. 

\section{Extension to other crossing-free structures}\label{sec:count-ocfs}

In this section we show how to count crossing-free structures other than triangulations. The idea is to use the framework first developed in~\cite{Alvarez:2012:CCS:2261250.2261259} that allows to exactly count crossing-free structures using the constrained Delaunay triangulation (\CDT) $\CDTT^{\CFS}$ constrained to contain the crossing-free structure $\CFS$. For completeness we will briefly describe the constrained Delaunay triangulation as explained in~\cite{Alvarez:2012:CCS:2261250.2261259,onion-layers}.

\begin{quote}

{\textbf{The constrained Delaunay triangulation}} (CDT) $\triangle^{S}$ of $\setp$ was first introduced in~\cite{Chew:1987:CDT:41958.41981}. Formally, it is the triangulation $T$ of $\setp$ containing $S$ such that no edge $e$ in $T\setminus S$ is flippable in the following sense: Let $\triangle_{1}, \triangle_{2}$ be triangles of $\setp$ sharing $e$. The edge $e$ is flippable if and only if $\square = \triangle_{1}\cup\triangle_{2}$ is convex, and replacing $e$ with the other diagonal of $\square$ increases the smallest angle of the triangulation of $\square$. One of the most important properties of constrained Delaunay triangulations is its uniqueness if no four points of $\setp$ are cocircular. Thus, under standard non-degeneracy assumptions, \emph{there is a unique CDT for any given set of mandatory edges}. For a good study on constrained Delaunay triangulations we suggest the book~\cite{Hjelle:2006:TA:1214284} by {\O}. Hjelle and M. D{\ae}hlen.
\end{quote}

Without loss of generality we will assume that no four points of $\setp$ are cocircular. This can always be achieved by perturbing $\setp$ while keeping the order-type of $\setp$.

Using the constrained Delaunay triangulation, algorithms to exactly count crossing-free matchings and spanning cycles (polygonizations) of $\setp$ were shown in~\cite{Alvarez:2012:CCS:2261250.2261259,onion-layers}. The main idea is the following: Assume we want to count all elements $\CFS \in \F_{\C}(\setp)$ of a certain class $\C$ of crossing-free structures on $\setp$. Now, instead of counting all $\CFS$ directly we count the pairs $(\CDTT^{S},S)$ for all $\CFS \in \F_{\C}(\setp)$, \emph{i.e.}, we count every element $\CFS \in \F_{\C}(\setp)$ embedded in its $\CDT$ $\CDTT^{S}$. This yields the correct number, since $\CDTT^{S}$ is unique w.r.t.\ $S$ by our non-degeneracy assumption. As we can see the embedding of $S$ in $\CDTT^{S}$ as annotating every edge of the triangulation $\CDTT^{S}$ by a bit specifying whether that edge belongs to $S$, we end up counting annotated triangulations. 
For this we can use similar ideas as in the case of counting triangulations: Having a set of separators $\SS$ for triangulations we enumerate all separators $C \in \SS$ and use the divide-and-conquer approach to count the crossing-free structures $\CFS$ such that $C\subseteq\CDTT^{S}$. For example, $\SS$ could be simple cycle separators, presented in~\S~\ref{sec:preliminaries}. It is important to observe that when recursing in the smaller sub-problems given by $C$, we have to have a way to locally verify whether $C\subseteq\CDTT^{S}$, \emph{i.e.}, we have to make sure that choices in a sub-problem do not depend on choices in other sub-problems, as otherwise we might get an overcounting by much more than the number of separators. Recall that we want to keep overcounting under control while achieving sub-exponential runtime.

In order to make the previous general idea clear we will adapt the way matchings are counted in~\cite{Alvarez:2012:CCS:2261250.2261259,onion-layers}. The main difference in the algorithms is the choice of separators. While the separators used in~\cite{Alvarez:2012:CCS:2261250.2261259,onion-layers} allow to count exactly, the simple cycle separators used in this paper are useful to count approximately, since $\CDTT^{S}$ might have more than one simple cycle separator.

\subsection{Counting matchings approximately}

Let $M$ be a matching of $\setp$, not necessarily perfect. Let $\CDTT^{M}$ be the $\CDT$ of $M$. As stated before, we are working under the assumption that no four points of $\setp$ are cocircular, thus $\CDTT^{M}$ is unique w.r.t.\ $M$. Now, annotate $\CDTT^{M}$ as follows:
\begin{itemize}
	\item Each vertex $v$ of $\CDTT^{M}$ is annotated with a number $m_{v}$ that indicates the vertex of~$M$ that $v$ is matched to. If $v$ is unmatched in $M$ we set $m_v$ to, say, $0$.
	\item Each edge $e$ of $\CDTT^{M}$ is annotated with a bit $b_{e}$ that indicates whether $e$ belongs to~$M$ or not.
\end{itemize}

Let us denote by $\CDTA^{M}$ the $\CDT$ $\CDTT^{M}$ annotated according to the previous rules. Let $C\in\SS$ be a separator contained in $\CDTT^{M}$ that splits $\Conv(\setp)$ into regions $R_{1},\ldots, R_{t}$. Separator $C$ inherits all the information, \emph{i.e.}, annotations, from $\CDTA^{M}$. The separator thus annotated will be denoted by $C_{\CDTA^{M}}$.

\subsubsection{The algorithm}\label{sec:algo-cdt}

An annotated triangulation is said to be \emph{legal} if and only if it is identical to $\CDTA^{M}$, for some matching $M$ of $\setp$. Now observe that there is a one-to-one correspondence between matchings of $\setp$ and legal annotated triangulations. Thus, instead of counting matchings directly we may count legal annotated triangulations. To do so we proceed essentially as in the algorithm presented in~\S~\ref{sec:count-tri}: (\oldstylenums{1}) We enumerate \emph{all} separators in~$\SS$. (\oldstylenums{2}) We enumerate \emph{all} annotations for each separator $C\in\SS$. Each such annotated separator splits $\Conv(\setp)$ into smaller regions we recur in. In each recursive sub-problem we count legal annotated triangulations that are consistent with the annotated separator, \emph{i.e.}, for example, if two adjacent vertices of the separator have been annotated, and they agree to be matched to each other and the edge connecting them is annotated to be in the matching, then in future sub-problems other edges adjacent to those two vertices cannot be annotated to be in a matching as well. (\oldstylenums{3}) We stop the recursion whenever we find a sub-problem of size at most $\SC$, just as in the algorithm shown in \S~\ref{sec:count-tri}, and compute the \emph{exact} number of legal annotated triangulations that are consistent with the annotations of the boundary of the sub-problem.

It should be clear by now that the only sub-problems that will contribute to the final computed number of matchings are the ones for which the algorithm, in its whole run, could complete a full annotated constrained Delaunay triangulation without finding any violation of the annotations inherited by the separators that led to that triangulation. There are however two details we have not yet taken care of: (\oldstylenums{1}) How do we verify that \emph{each} edge that is not annotated to be in a matching $M$ is not flippable in $\CDTA^{M}$? As it stands right now, for an annotated separator $C$ dividing $\Conv(\setp)$ into regions $R_{1}, \ldots, R_{t}$, we recur into each $R_{i}$, $1\leq i\leq t$ independently, although an edge $e$ of $C$ could become flippable depending the two triangles containing $e$ that lie in different regions. (\oldstylenums{2}) How do we \emph{exactly} count legal annotated triangulations that are consistent with the annotations of the boundary of a sub-problem of size at most $\SC$? We will tackle each difficulty in turn.

\subsubsection{Triangular cycle separators}

To overcome difficulty (\oldstylenums{1}) mentioned before, we make use of what we can call ``fat'' separators, which are very similar to a construction in~\cite{Alvarez:2012:CCS:2261250.2261259, onion-layers}. Let $T$ be a triangulation of point set $P$ and let $C \subset T$ be a cycle separator. Every edge of $C$ that belongs to the interior of $\Conv(\setp)$ is contained in two triangles in $T$. For every such edge choose those two triangles, and for every edge of $C$ on the boundary of $\Conv(\setp)$ choose the unique triangle that the edge is part of. We make $C$ fat by considering the union of all those chosen triangles (as a set of edges), see Figure~\ref{figs:fig1}. We call the resulting fat version of the cycle separator $C$ a \emph{triangular cycle separator} and denote it by $C^{\circ}$.

\begin{figure}[!htb]
\begin{center}
\includegraphics[height=5cm]{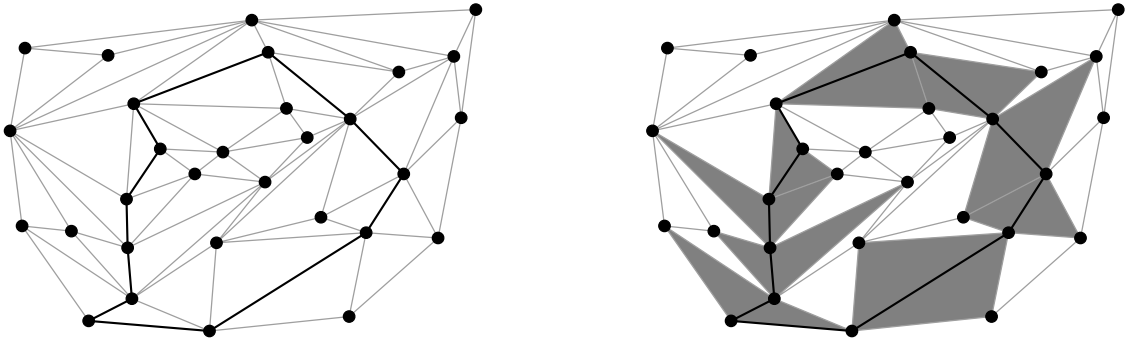}
\caption{To the left a triangulation and a simple cycle, shown in gray and black respectively. To the right the cycle is made ``fat'' by attaching the shown triangles to it.}
\label{figs:fig1}
\end{center}
\end{figure}

Fat separators allow us to overcome the difficulty regarding edge-flippability as follows. Instead of using simple cycle separators, as we did in~\S~\ref{sec:count-tri}, we use triangular cycle separators for counting matchings. Since for every edge $e$ contained in a fat separator $C^\circ$ at least one of its incident triangles is fully contained in $C^\circ$, we can check flippability of $e$ in one of the sub-problems, namely the one containing the other triangle incident to $e$. Thus, flippability can be checked independently in each sub-problem.

\subsubsection{Counting legal annotated triangulations exactly}\label{sec:count-cdts}

Assume that at a certain recursion level we have finally reached a sub-problem $\SP$ of size at most $\SC$. The boundary of this sub-problem has been annotated and now we have to exactly count \emph{all} legal annotated triangulations that are consistent with these annotations. In order to achieve this we make use of enumeration algorithms for triangulations, see~\cite{Avis199621,Bespamyatnikh2002271,Katoh:2009aa}. For each enumerated triangulation $T = T(\SP)$ we consider \emph{each} sub-set $E^{\prime}\subseteq E = E(T)$ of its set of edges $E$. 
For each such $E'$ we check whether $E'$ forms a valid matching, is consistent with the annotations of the boundary of $\SP$, and violates no flippability condition. We count all sets $E'$ satisfying these conditions. It should be clear that using this brute-force approach we are exactly counting \emph{all} legal annotated triangulations that are consistent with the annotations of $\SP$. 

The number of triangulations of $\SP$ can be expressed as $c^{\SC}$, for some positive constant~$c$. The number of edges in each triangulation of $\SP$ is $O(\SC)$. Thus, enumerating every sub-set $E'$ of edges takes time $2^{O(\SC)}$. Verifying whether the chosen sub-set of edges forms a consistent matching and checking for flippability of the edges can easily be done within the same time bound. Hence, overall we can solve sub-problems of size $\SC$ in time $2^{O(\SC)}$.

\subsubsection{Quality of approximation and running time}\label{sec:count-cdts:proof}

Having overcome the two difficulties we mentioned previously, we are now left with the analysis of the approximation quality and runtime of the algorithm. Fortunately, essentially all the work has been done in~\S~\ref{sec:theoremTri}.

From Lemma~\ref{lemmas:1} we obtain that the total number of simple cycle separators is $n^{O(\sqrt{n})}$. Now, we make each such simple cycle separator fat by attaching at most two triangles to each of its edges. It is not hard to see that the total number of fat simple cycle separators, or triangular cycle separators, is $n^{O(\sqrt{n})}$, just observe that for every edge of a simple cycle separator we choose two additional points that will complete the attached triangles to make it fat. 
Next consider the annotations of the triangular cycle separators. It is easy to verify that the number of different annotations for the vertices of a separator is at most $n^{O(\sqrt{n})}$ and the number of annotations for its edges is at most $2^{O(\sqrt{n})}$. Therefore, the total number of \emph{annotated} triangular cycle separators is $n^{O(\sqrt{n})}$.

Now, observe that for the analysis in~\S~\ref{sec:theoremTri} we essentially only used the following three facts. (\oldstylenums{1}) The number of (annotated) separators of each sub-problems is at most $n^{O(\sqrt{n})}$. (\oldstylenums{2}) A sub-problem of size $\SC$ can be solved in time $2^{O(\SC)}$. (\oldstylenums{3}) All further checks can be performed in polynomial time. All three facts still hold in the case of matchings. Hence, the analysis from~\S~\ref{sec:theoremTri} goes through and we obtain the same asymptotic bounds for approximation quality and runtime. This finishes the analysis of the algorithm for counting matchings approximately. 

\subsection{Counting trees approximately}

Having shown how to count matchings approximately using annotations similar to~\cite{Alvarez:2012:CCS:2261250.2261259,onion-layers}, we now show an annotation scheme that allows to count (crossing-free) spanning trees of~$\setp$ approximately. 
Fixing an arbitrary vertex $p^* \in P$ we root every spanning tree $F$ at $p^*$ by orienting all of its edges towards $p^*$. This way, every point $p$ in $P \setminus \{p^*\}$ has exactly one outgoing edge, and we call the other end of that edge the \emph{parent} of $p$ in~$F$, denoted by $\textup{par}(p)$ (for $p^*$ we set $\textup{par}(p^*) := \nil$). Moreover, we denote by $d(p)$ the depth of $p$ in $F$, \emph{i.e.}, its distance to $p^*$ in $F$.
Note that the resulting oriented spanning trees of $P$ rooted at $p^*$ are in one-to-one correspondence to (ordinary) spanning trees of $P$, so we may count the former. In the following we will write for short \emph{spanning tree} for oriented spanning trees of $P$ rooted at~$p^*$. We annotate the CDT $\CDTT^{F}$ of a spanning tree $F$ of $P$ as follows:

\begin{itemize}
	\item Every vertex $v$ of $F$ is annotated by the pair $<\!\!\textup{par}(v),d(v)\!\!>$.
	\item Every edge $e$ of $\CDTT^{F}$ is annotated with a bit $b_{e}$ that again represents whether $e$ is part of $F$ or not.
\end{itemize}

As for matchings, the annotated CDT $\CDTT^{F}$ will be denoted by $\CDTA^{F}$. Again there is a one-to-one correspondence between (oriented, rooted at $p^*$) spanning trees of $\setp$ and legal annotated triangulations. Therefore, we will again count the latter, just as we did with matchings. In fact, the algorithm is the same as in~\ref{sec:algo-cdt}, the only difference is annotation scheme that encodes a different class of crossing-free structures.

Nevertheless, in the case of spanning trees two things that are less obvious: (\oldstylenums{1}) How to exactly count all legal annotated triangulations of sub-problems of size at most $\SC$, and (\oldstylenums{2}) the correctness of the annotation scheme: why do we only (over-)count spanning trees? Let us start with the former.

Assume that we arrive at a sub-problem $\SP$ of size at most $\SC$. In $\SP$ there are already some annotations present and we have to count \emph{all}  annotated triangulations that are consistent with those annotations. We will proceed as before: We enumerate each triangulation $T = T(\SP)$ of $\SP$ and every sub-set $E^{\prime}\subseteq E = E(T)$ of edges of $T$ that have no annotation set yet. All edges of $E^{\prime}$ are chosen to be in the spanning tree. At this point we can check for flippability conflicts. If we find any conflict we abort and consider the next sub-set of edges, otherwise we continue. Since we are interested in (oriented, rooted at $p^*$) spanning trees, we enumerate \emph{every} possible orientation of the edges in $E'$, observe that there are at most $2^{O(|E^{\prime}|)} = 2^{O(\SC)}$ possible orientations we have to go through. This orientation sets the parent of every vertex of the edges of $E^{\prime}$. Recall that we want to orient towards the root $p^*$. Now we test whether this orientation of the edges of $E^{\prime}$ is consistent with the annotations already present in $\SP$, \emph{e.g.}, a vertex $v$ that was already annotated in $\SP$, because it belonged to some separator, must have its parent set. If this parent is in $\SP$, then $v$ and its parent must be vertices of some edge of $E^{\prime}$, otherwise we can abort this run. This in particular tests that the \emph{unique} parent of every vertex of $T$ is set correctly if this parent belongs to $\SP$ as well.

\begin{figure}[!htb]
\begin{center}
\includegraphics[height=5cm]{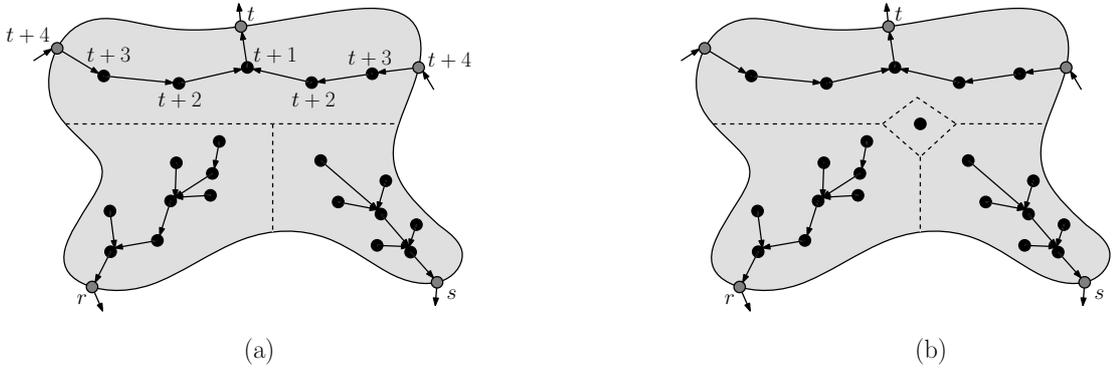}
\caption{A sub-problem $\SP$ of size at most $\SC$ is shown in gray. The boundary of $\SP$ is shown simplified, \emph{i.e.}, no triangular cycle separator is shown. The set of edges $E^{\prime}$ chosen to be in a tree is shown along with an orientation of its elements. For simplicity, the triangulation of $\SP$ that the set $E^{\prime}$ is chosen from is not shown, so in particular we are assuming that there is no flippability conflict. In (a) there are three connected components $F_{1}, F_{2}, F_{3}$, separated by dashed lines for clarity. The gray vertices on the boundary of $\SP$ are fully annotated, so we can start a BFS for each component at any of those vertices. The labels $r, s, t$ represent the annotated depth of those vertices. Starting with the vertex at depth $t$, every vertex of that component gets the shown depth assigned by the BFS. By the time we arrive at the vertices of depth $t + 4$ on the boundary of $\SP$, since they are fully annotated, we compare $t + 4$ with the depth therein annotated and verify for consistency. If they are inconsistent then we abort the run with this set of orientations to the edges of $E^{\prime}$. In (b) there are four connected components but the single vertex in the middle will never be annotated since there is no edge in $E^{\prime}$ that has it as a vertex, so even if we found consistent annotations for the other components, this run is aborted.}
\label{figs:fig2}
\end{center}
\end{figure}

Having chosen $E^{\prime}$ and having set orientations for its elements, note that we are left with some set of connected components $F_{1},\ldots F_{k}$ formed by the (oriented) edges of $T$ that we have chosen to be in the tree, see Figure~\ref{figs:fig2}.a. Now, for each connected component $F_{i}$, $1\leq i\leq k$, let $v^{i}$ denote an arbitrary vertex that was found fully annotated when first entering $\SP$, any vertex of the boundary of $\SP$ could be such a vertex since they are all fully annotated. Perform a Breadth-first search (BFS) of $F_{i}$ starting at $v^{i}$ as if the edges of $F_{i}$ had no orientations. Let $v$ be the current vertex the BFS is about to expand its neighborhood of, so in the beginning $v = v^{i}$, and let $d_{v}$ be its depth. Set the depth of every vertex $w$ of $F_{i}$ that was found while expanding the neighborhood of $v$ to $d_{v} + 1$ or $d_{v} - 1$ depending on the orientation of the edge $e = vw$, \emph{i.e.}, if $e$ is oriented towards $w$, then $w$ is the parent of $v$ and thus we set its depth to $d_{v} - 1$, and if $e$ is oriented towards $v$ we set it to $d_{v} + 1$. If the depth of $w$ is found to be already set when doing the BFS, then we just verify that the value the algorithm would write there and the value already there coincide. If those two numbers do not coincide, then we abort and discard this set of orientations given to the edges of $E^{\prime}$, otherwise we continue with the BFS. If the whole connected component $F_{i}$ was successfully traversed, then we manage to find a consistent set of annotations for it and we can continue with another connected component, if any. If after considering every connected component there is a vertex of $T$ that has not been fully annotated, then we dismiss this run as well, since we were not able to find a consistent set of annotations for $E^{\prime}$ in $\SP$, see Figure~\ref{figs:fig2}.b, otherwise we declare this run a success.

With this method we clearly exactly count \emph{all} legal annotated triangulations that are consistent with the annotations we started with for sub-problem $\SP$. It remains to show that, overall, \emph{only} spanning trees of $\setp$ are (over-)counted.

\begin{lemma}
	Let $\setp$ be a set of $n$ points. Then our annotation scheme encodes the spanning trees of $\setp$ unequivocally.
\end{lemma}
\begin{proof}
Consider a legal annotated triangulation $\overline{T}$ of $P$. The annotations $b_e$ of its edges induce a set of edges $F$ such that $\overline{T}$ is the CDT $\CDTT^{F}$ of $F$ (by flippability constraints). Moreover, the edges in $F$ are consistent with the annotations $<\!\!\textup{par}(v),d(v)\!\!>$ of the vertices of $\overline{T}$, and the latter give $F$ an orientation. Thus, every vertex except for the root $p^*$ has a unique parent, so it has out-degree 1, and $F$ consists of exactly $n-1$ edges. It remains to show that $F$ has no undirected cycle, then the annotations induce an (oriented, rooted at $p^*$) spanning tree of $P$. For this, first note that any undirected cycle in $F$ would also be a directed cycle  (according to the orientation of $F$ given by the parent annotations), since all vertices have out-degree at most 1. Moreover, consistency of the annotated depths $d(v)$ implies that $d(\textup{par}(v)) = d(v) - 1$ for all vertices $v$. Hence, we also cannot have any directed cycle in $F$, as otherwise for at least one edge of the cycle the depths cannot be consistent.
This finishes the proof.
\end{proof}

Finally, it remains to show what approximation factor and running time we obtain with this algorithm. For the former, the explanation given in~\S~\ref{sec:count-cdts:proof} follows verbatim. For the running time it follows almost verbatim, the only difference is that when choosing a sub-set of edges $E^{\prime}$ that will be part of a spanning tree, we have to orient them as well. This orientation comes with a $2^{|E^{\prime}|} = 2^{O(\SC)}$ overhead. That is, instead of having a runtime of $T(\SC) = c^{\SC}\cdot 2^{O(\SC)} = 2^{O(\SC)}$, as we had when counting matchings in~\ref{sec:count-cdts:proof}, we now have a runtime of~$2^{O(\SC)}\cdot~T(\SC) = 2^{O(\SC)}$, which is asymptotically the same. The time required to perform the breath-first searches also have no effect in asymptotic terms. Thus, choosing again $\SC = \sqrt{n}\log(n)$ gives us an overall running time of $T(n) = 2^{O\left(\sqrt{n}\log(n)\right)}$. Hence, both algorithm have the same asymptotic behavior.

\section{Conclusions}\label{sec:conclusions}

In this paper we have shown algorithms that, given a set of points, count triangulations, crossing-free matchings and crossing-free spanning trees approximately. Both, the approximation ratio and the running time are sub-exponential. We would like to remark that using the set of annotations for crossing-free spanning cycles shown in~\cite{Alvarez:2012:CCS:2261250.2261259} we can modify the algorithms presented in this paper to also count crossing-free spanning cycles approximately, again with sub-exponential approximation ratio and running time. After showing the algorithms for matchings and spanning trees, this extension is straightforward. On the other hand, the annotations shown in this paper to count spanning trees were not known before. They are compatible with the algorithm of~\cite{Alvarez:2012:CCS:2261250.2261259}, and combining the two yields an algorithm to count crossing-free spanning trees \emph{exactly} in time $n^{O(k)}$, where $k$ is the number of onion layers of the given set of points. Specifically, as long as $k$ is fixed we can exactly count spanning trees in polynomial time, which is an interesting result in itself.

Finally, we can express the cardinality of any of the classes of crossing-free structures considered in this paper as $c^{n}$, where $c$ depends on the given set of points but is sandwiched between two positive constants. Although the approximation ratios shown in this paper are rather large, our algorithms compute a $(1+o(1))$-approximation of the base $c$, and it does so in sub-exponential time. No algorithm with this property was known before. However, it remains an open problem to find an algorithm with sub-exponential approximation ratio and running time $2^{o(\sqrt{n} \log(n))}$, \emph{e.g.}, polynomial.


\bibliographystyle{abbrv}
\bibliography{CountingTriangulationsApproximatelyArXiv}

\end{document}